%% file: root.tex
\documentclass[11pt]{article}

\usepackage[utf8]{inputenc}
\usepackage[T1]{fontenc}
\usepackage[sc]{mathpazo}
\usepackage{microtype}
\usepackage{listings}
\usepackage{amsmath}
\usepackage{amssymb}
\usepackage{amsthm}
\usepackage{dsfont}
\usepackage[noblocks]{authblk}
\usepackage{fullpage}
\usepackage{parskip}
\usepackage{comment}
\usepackage{tikz}
\usepackage{mathtools}
\usepackage[shortlabels]{enumitem}
\usepackage[ruled]{algorithm2e}
\usepackage{complexity}
\usepackage[backend=biber, style=alphabetic, maxbibnames=99]{biblatex}
\usepackage[breaklinks]{hyperref}
\usepackage[nameinlink]{cleveref}
\usepackage{nag}

\begingroup
\makeatletter
\@for\theoremstyle:=definition,remark,plain\do{%
    \expandafter\g@addto@macro\csname th@\theoremstyle\endcsname{%
        \addtolength\thm@preskip\parskip
        }%
    }
\endgroup

\newcommand{\declarecolor}[2]{\definecolor{#1}{RGB}{#2}\expandafter\newcommand\csname #1\endcsname[1]{\textcolor{#1}{##1}}}
\declarecolor{White}{255, 255, 255}
\declarecolor{Black}{0, 0, 0}
\declarecolor{LightGray}{216, 216, 216}
\declarecolor{Gray}{127, 127, 127}
\declarecolor{Orange}{237, 125, 49}
\declarecolor{LightOrange}{251,229, 214}
\declarecolor{Yellow}{255, 192, 0}
\declarecolor{LightYellow}{255, 242, 200}
\declarecolor{Blue}{91, 155, 213}
\declarecolor{LightBlue}{222, 235, 247}
\declarecolor{Green}{112, 173, 71}
\declarecolor{LightGreen}{226, 240, 217}
\declarecolor{Navy}{68, 114, 196}
\declarecolor{LightNavy}{218, 227, 243}

\hypersetup{
	colorlinks=true,
	pdfpagemode=UseNone,
	citecolor=Green,
	linkcolor=Navy,
	urlcolor=Black,
	pdfstartview=FitW
}

\newcommand{\declareperson}[1]{\expandafter\newcommand\csname#1\endcsname[1]{\textcolor{Orange}{#1: ##1}}}

\declareperson{Nima}
\declareperson{Alireza}

\theoremstyle{plain}
\newtheorem{theorem}{Theorem}
\newtheorem{lemma}[theorem]{Lemma}

\newtheorem{fact}[theorem]{Fact}

\theoremstyle{definition}
\newtheorem{definition}[theorem]{Definition}

\theoremstyle{remark}

\newlist{parts}{enumerate}{10}
\setlist[parts]{label=\arabic*.,ref=\arabic*}
\makeatletter
\if@cref@capitalise
	\crefname{partsi}{Part}{Parts}
\else
	\crefname{partsi}{part}{parts}
\fi
\makeatother
\Crefname{partsi}{Part}{Parts}

\setlist[itemize]{noitemsep}

\usetikzlibrary{shapes, patterns, decorations, fit, intersections}

\newcommand*{\Z}{{\mathbb{Z}}}
\let\R\relax
\newcommand*{\R}{{\mathbb{R}}}

\let\C\relax
\newcommand*{\C}{{\mathbb{C}}}
\newcommand*{\B}{{\mathcal{B}}}
\newcommand*{\1}{{\mathds{1}}}
\newcommand*{\SharpP}{{\ComplexityFont{\#P}}}
\newcommand*{\qbar}{\overline{q}}
\newcommand*{\sbar}{\overline{s}}

\let\poly\relax
\DeclareMathOperator{\poly}{poly}
\DeclareMathOperator{\per}{per}
\DeclareMathOperator{\bethe}{bethe}
\DeclareMathOperator{\bp}{bp}
\DeclareMathOperator{\argmax}{argmax}

\providecommand{\given}{}
\providecommand{\vs}{}
\DeclarePairedDelimiter{\bitcomplexity}\langle\rangle
\DeclarePairedDelimiterX{\card}[1]{\lvert}{\rvert}{\renewcommand\given{\nonscript\:\delimsize\vert\nonscript\:\mathopen{}}#1}
\DeclarePairedDelimiterX{\abs}[1]{\lvert}{\rvert}{\renewcommand\given{\nonscript\:\delimsize\vert\nonscript\:\mathopen{}}#1}
\DeclarePairedDelimiterX{\norm}[1]{\lVert}{\rVert}{\renewcommand\given{\nonscript\:\delimsize\vert\nonscript\:\mathopen{}}#1}
\DeclarePairedDelimiterX{\tuple}[1]{\lparen}{\rparen}{\renewcommand\given{\nonscript\:\delimsize\vert\nonscript\:\mathopen{}}#1}
\DeclarePairedDelimiterX{\parens}[1]{\lparen}{\rparen}{\renewcommand\given{\nonscript\:\delimsize\vert\nonscript\:\mathopen{}}#1}
\DeclarePairedDelimiterX{\brackets}[1]{\lbrack}{\rbrack}{\renewcommand\given{\nonscript\:\delimsize\vert\nonscript\:\mathopen{}}#1}
\DeclarePairedDelimiterX{\set}[1]\{\}{\renewcommand\given{\nonscript\:\delimsize\vert\nonscript\:\mathopen{}}#1}
\let\Pr\relax
\DeclarePairedDelimiterXPP{\Pr}[1]{\mathbb{P}}[]{}{\renewcommand\given{\nonscript\:\delimsize\vert\nonscript\:\mathopen{}}#1}
\DeclarePairedDelimiterXPP{\PrX}[2]{\mathbb{P}_{#1}}[]{}{\renewcommand\given{\nonscript\:\delimsize\vert\nonscript\:\mathopen{}}#2}
\DeclarePairedDelimiterXPP{\Ex}[1]{\mathbb{E}}[]{}{\renewcommand\given{\nonscript\:\delimsize\vert\nonscript\:\mathopen{}}#1}
\DeclarePairedDelimiterXPP{\ExX}[2]{\mathbb{E}_{#1}}[]{}{\renewcommand\given{\nonscript\:\delimsize\vert\nonscript\:\mathopen{}}#2}
\DeclarePairedDelimiterXPP{\KL}[1]{\mathcal{D}}(){}{\renewcommand\vs{\nonscript\:\delimsize\|\nonscript\:\mathopen{}}#1}

\newcommand*{\eval}[1]{\left.#1\right\rvert}

\title{A Tight Analysis of Bethe Approximation for Permanent}

\author{Nima Anari}
\affil{\small Stanford University, \textsf{anari@cs.stanford.edu}}

\author{Alireza Rezaei}
\affil{\small University of Washington, \textsf{arezaei@cs.washington.edu}}

\bibliography{refs}

\begin{document}
	\maketitle
	
	\begin{abstract}
		We prove that the permanent of nonnegative matrices can be deterministically approximated within a factor of $\sqrt{2}^n$ in polynomial time, improving upon previous deterministic approximations. We show this by proving that the Bethe approximation of the permanent, a quantity computable in polynomial time, is at least as large as the permanent divided by $\sqrt{2}^{n}$. This resolves a conjecture of \textcite{Gur11}. Our bound is tight, and when combined with previously known inequalities lower bounding the permanent, fully resolves the quality of Bethe approximation for permanent. As an additional corollary of our methods, we resolve a conjecture of \textcite{CY13}, proving that fractional belief propagation with fractional parameter $\gamma=-1/2$ yields an upper bound on the permanent.
	\end{abstract}
	
	\input{intro}

	\input{prelim}
	
	\input{bethe}
	
	\input{upperbound}

	\input{phi}

	\input{bp}
		
	\printbibliography
\end{document}

%% file: intro.tex
\section{Introduction}\label{sec:intro}

The \emph{permanent} of a square matrix has been an important object of study in combinatorics and theoretical computer science. The permanent of an $n\times n$ matrix $A$ is defined as
	\[ \per(A)=\adjustlimits\sum_{\sigma\in S_n} \prod_{i=1}^n A_{i, \sigma(i)}, \]
where $S_n$ is the set of all permutations $\sigma:\set{1,\dots,n}\to \set{1,\dots,n}$. Although this definition of the permanent looks superficially similar to that of the determinant, the latter is polynomial-time computable while the former is \SharpP-hard to compute even when $A$ has entries in $\set{0,1}$ \cite{Val79}.

Given the hardness of computing the permanent exactly, a natural question is when can the permanent be approximated efficiently? For arbitrary $A\in \C^{n\times n}$, an $\epsilon \norm{A}^n$ additive approximation\footnote{We remark that for any matrix $A$, $\abs{\per(A)}\leq \norm{A}^n$.} was given by \textcite{Gur05}; this was later improved and derandomized in certain cases \cite{AH14, CCGP17}. The permanent of complex matrices is a central quantity in quantum linear optics, and is connected to some methods proposed for demonstrating ``quantum supremacy'' \cite{AA11}.

A \emph{multiplicative} approximation is elusive for general $A$; even determining the sign of the permanent of real matrices is hard \cite{AA11, GS16}. This motivates studying classes of matrices where the sign is not a barrier. Polynomial time multiplicative approximation algorithms have been designed for two large classes known to have a nonnegative permanent: the most prominent is a randomized $(1+\epsilon)$-approximation algorithm for matrices with nonnegative entries \cite{JSV04}, and there is also a $\simeq 4.84^n$-approximation algorithm for positive semidefinite matrices \cite{AGOS17}. We remark that there are cases where the permanent is not even necessarily real-valued, yet quasi-polynomial time multiplicative $(1+\epsilon)$-approximation has been achieved  \cite{Bar16, Bar18, EM17}. In this paper we focus only on matrices with nonnegative entries; for these, the permanent has a close connection to the combinatorics of perfect matchings in bipartite graphs.

For $A\in \set{0,1}^{n\times n}$, the permanent of $A$ is equal to the number of perfect matchings in the bipartite graph with bipartite adjacency matrix $A$. For nonnegative $A\in \R_{\geq 0}^{n\times n}$ the permanent is clearly a nonnegative quantity, and in fact can be thought of as a partition function. Consider sampling $\sigma \in S_n$ with probability $\Pr{\sigma}\propto \prod_{i=1}^n A_{i, \sigma(i)}$. Then $\per(A)$ is the normalizing constant in this model, i.e., $\Pr{\sigma}=\parens{\prod_{i=1}^n A_{i, \sigma(i)}}/\per(A)$. When $A\in \set{0,1}^{n\times n}$, this distribution over $\sigma\in S_n$ is the same as the uniform distribution over perfect matchings in the corresponding bipartite graph.

Given the close connection between random perfect matchings and the permanent of nonnegative matrices, \textcite{Bro86} suggested using Monte Carlo Markov Chain (MCMC) methods to approximate the permanent. This was finally achieved by \textcite{JSV04}, who gave a fully polynomial time randomized approximation scheme (FPRAS); that is a randomized algorithm that for any desired $\epsilon,\delta>0$, with probability $1-\delta$ computes a $1+\epsilon$ approximation to $\per(A)$ for nonnegative $A$ in time $\poly(n, \bitcomplexity{A}, 1/\epsilon,\log(1/\delta))$, where $\bitcomplexity{A}$ represents the bit complexity of $A$.

\subsection{Deterministic approximation of the permanent}

Despite having a complete understanding of approximations achievable by randomized algorithms, there is currently a large gap in our understanding of deterministic permanent approximation. The only known lower bound is still the \SharpP-hardness of exact computation \cite{Val79}. On the other hand, the best known deterministic approximation ratio is $\simeq 1.9022^n\leq 2^n$ \cite{GS14}. 

Deterministic approximations of the permanent can sometimes be useful in optimization tasks. For example, \textcite{AOSS17} used a particular approximation of the permanent to attack the problem of maximizing Nash social welfare; a key property used in their algorithm was that this approximation of $\per(A)$ was log-concave as a function of the matrix $A$, so it could be maximized for $A$ ranging over a convex set of matrices. Another problem that can be cast as permanent maximization is that of finding the profile maximum likelihood, a successful universal method for estimating symmetric statistics from sampled data \cite[see ][]{ADOS17}. One of the successful heuristics used to find the profile maximum likelihood estimate is based on a particular approximation of the permanent called the Bethe approximation \cite{Von14}. In fact this is the same approximation we study in this paper. For more recent results on profile maximum likelihood see \cite{ADOS17,PJW17,CSS18}.

It is no coincidence that the gap between randomized and deterministic approximations of permanent is so large: $1+\epsilon$ vs.\ exponential. This is because, roughly speaking, no approximation ratio can be in between the two extremes.\footnote{We remark that this does not rule out the possibility of the best approximation ratio being for example $2^{n/\poly\log(n)}$.} More precisely, any polynomial time $2^{n^c}$-approximation for any constant $c<1$ can be boosted to obtain a $(1+\epsilon)$-approximation. This is achieved by feeding the hypothetic algorithm the larger matrix
\[ I_{m}\otimes A=\begin{bmatrix}
	A & 0 & \dots & 0\\
	0 & A & \dots & 0\\
	\vdots & \vdots & \ddots & \vdots \\
	0 & 0 & \dots & A\\
\end{bmatrix},\]
whose size is $mn\times mn$ and whose permanent is $\per(A)^m$. By taking the $m$-th root of the hypothetic algorithm's approximation of $\per(I_m\otimes A)$, we obtain a $\sqrt[m]{2^{(mn)^c}}$-approximation to $\per(A)$. It is enough to take $m\simeq (n^c/\epsilon)^{1/(1-c)}=\poly(n, 1/\epsilon)$ for this approximation ratio to become $1+\epsilon$.

This suggests that the right question to ask about the approximation ratio is: What is the infimum of constants $\alpha>1$ for which the permanent can be deterministically approximated within a factor of $\alpha^n$ in polynomial time? \Textcite{LSW00} achieved the first nontrivial deterministic approximation of the permanent and showed that $\alpha\leq e$. Later \textcite{GS14} showed\footnote{The approximation ratio obtainable by the analysis of \textcite{GS14} is actually $\simeq 1.9022^n$ but this is not explicitly mentioned in their article, and the only explicit bound mentioned is $2^n$.} that $\alpha\leq 1.9022<2$. We prove that $\alpha\leq \sqrt{2}$.
\begin{theorem}\label{thm:main}
	There is a deterministic polynomial time algorithm that approximates the permanent of nonnegative matrices within a factor of $\sqrt{2}^n$.
\end{theorem}
We prove this result by improving the analysis of an already well-studied algorithm. The  factor of $\sqrt{2}^n$ for the output of this algorithm was already known to be achieved for a class of matrices, and was conjectured to be right approximation factor \cite{Gur11}.

\subsection{Bethe approximation of the permanent}

The algorithm we analyze has, in fact, been studied very extensively and its output is known as the Bethe approximation of the permanent, or Bethe permanent for short \cite{Von13}.
\begin{definition}\label{def:BP}
	For a matrix $A\in \R_{\geq 0}^{n\times n}$, define the Bethe permanent of $A$ as
	\[ \bethe(A)=\max \set*{\prod_{i,j=1}^n \parens*{A_{i,j}/P_{i, j}}^{P_{i,j}}\parens*{1-P_{i,j}}^{1-P_{i, j}} \given
		\begin{array}{l}
			P\in \R_{\geq 0}^{n\times n},\\
			P\1 = \1,\\
			\1^\intercal P =\1^\intercal.
		\end{array}}
	\]
\end{definition}
This maximization is defined over \emph{doubly stochastic} matrices $P$; that is nonnegative matrices whose rows and columns sum to $1$.

Bethe approximation is a heuristic used to approximate the partition function of arbitrary graphical models, whose roots go back to statistical physics. Bethe approximation is very closely connected to the sum-product algorithm, also known as belief propagation \cite{YFW05, Von13}.

Natural graphical models can be easily constructed where the permanent is the partition function and the corresponding Bethe approximation, the Bethe permanent. It was observed that the sum-product algorithm on these graphical models yields a good approximation to the permanent in practice \cite{CKV08, HJ09}. Later \textcite{Von13} showed that for these graphical models the sum-product algorithm always converges to the maximizer in \cref{def:BP}, and that the maximization objective in \cref{def:BP} becomes a concave function of $P$ after taking $\log$. So besides the sum-product algorithm, $\bethe(A)$ can also be computed by convex programming methods. We remark that this is one of the rare instances where the sum-product algorithm is \emph{guaranteed} to converge.

The first theoretical guarantee for the quality of Bethe approximation was actually obtained by \textcite{GS14}. Their $2^n$-approximation guarantee was for a slightly modified variant of the Bethe permanent. In particular, \textcite{Gur11} showed that $\bethe(A)$ always lower bounds $\per(A)$ and later \textcite{GS14} showed that $2^n \bethe(A)$ upper bounds $\per(A)$.  The lower bound was proved by using an inequality of \textcite{Sch98}. Two alternative proofs of the lower bound have since been obtained based on 2-lifts \cite{Csi14}, and real stable polynomials \cite{AO17}. For a survey of these proofs and their generalizations beyond permanent see \textcite{SV17}.
\begin{theorem}[\cite{Gur11} based on \cite{Sch98}]\label{thm:lowerbound}
	For any matrix $A\in \R_{\geq 0}^{n\times n}$ we have
	\[ \per(A)\geq \bethe(A). \]
\end{theorem}
While the lower bound is an elegant inequality that is also tight \cite{GS14}, the upper bound of $2^n\bethe(A)$ had room for improvement. In fact \textcite{Gur11} had already conjectured that the correct constant for the upper bound should be $\sqrt{2}^n$. We resolve this conjecture.
\begin{theorem}\label{thm:upperbound}
	For any matrix $A\in \R_{\geq 0}^{n\times n}$ we have
	\[ \per(A)\leq \sqrt{2}^n \bethe(A). \]
\end{theorem}
\Cref{thm:upperbound,thm:lowerbound} together with polynomial time computability \cite{Von13, GS14} imply \cref{thm:main}.

\subsection{Fractional belief propagation}

As an additional corollary of our methods, we prove an upper bound on the permanent conjectured by \textcite{CY13}. They proposed a way to interpolate between the mean-field and the Bethe approximations of the permanent, which result in, respectively, upper and lower bounds for the permanent.
\begin{definition}[\cite{CY13}]\label{def:bp}
	For a matrix $A\in \R_{\geq 0}^{n\times n}$ and parameter $\gamma\in \R$, define the $\gamma$-fractional belief propagatation approximation for the permanent of $A$, which we denote by $\bp_\gamma(A)$, as
	\[ \bp_\gamma(A)=\max \set*{\prod_{i,j=1}^n (A_{i,j}/P_{i,j})^{P_{i,j}}(1/(1-P_{i,j}))^{\gamma (1-P_{i,j})} \given \begin{array}{l}
			P\in \R_{\geq 0}^{n\times n},\\
			P\1 = \1,\\
			\1^\intercal P =\1^\intercal.
		\end{array}} \]
\end{definition}
Note that for $\gamma=-1$, we get $\bp_\gamma=\bethe$. Furthermore, for $\gamma=1$, we obtain what is classically known as the mean-field approximation. The former is a lower bound on the permanent, and the latter is an upper bound. In fact, for any $\gamma\geq 0$, we have $\bp_\gamma(A)\geq \per(A)$ \cite{CY13}. In both cases, the objective in \cref{def:bp} is log-concave as a function of $P$. \Textcite{CY13} observed that for any $\gamma\in[-1,1]$, the objective in \cref{def:bp} is also log-concave in terms of $P$, because it is a geometric mean of the $\gamma=-1$ and $\gamma=1$ cases, and thus $\bp_\gamma(A)$ can be computed efficiently for any $\gamma\in [-1,1]$.

\Textcite{CY13} also observed that $\bp_\gamma$ is monotonically increasing in terms of $\gamma$, and thus one can study the parameter $\gamma^*(A)$ for which $\bp_{\gamma*(A)}(A)=\per(A)$; to resolve corner cases where an interval of such parameters exists, consider the least parameter $\gamma^*$ for which this happens. They showed, based on previously known inequalities, that $\gamma^*(A)\leq 0$, but they experimentally observed that $\gamma^*(A)\leq -1/2$, and conjectured this to be always the case. We prove this conjecture.
\begin{theorem}\label{thm:bp}
	For any matrix $A\in \R_{\geq 0}^{n\times n}$ we have
	\[ \per(A)\leq \bp_{-1/2}(A). \]
\end{theorem}

We remark that the permanent upper bounds in \cref{thm:bp,thm:upperbound} are, in general, incomparable. However, the ratio between $\bp_{-1/2}(A)$ and $\bp_{-1}(A)=\bethe(A)$ is at most $\sqrt{e}^n$. Thus $\bp_{-1/2}(A)$, which is efficiently computable, provides a $\sqrt{e}^n$ deterministic approximation to the permanent. See \cref{sec:bp} for more details. For comparison, $1.9022>\sqrt{e}\simeq 1.6487>\sqrt{2}$.

\subsection{Techniques for upper bounding permanent}

Perhaps the most well-known upper bound on the permanent is the Bregman-Minc inequality, conjectured by \textcite{Min63} and proved by \textcite{Bre73}, which states that for $A\in \set{0,1}^{n\times n}$ if the row sums are $d_1,\dots, d_n$, then
\[ \per(A)\leq \sqrt[d_1]{d_1!}\cdots \sqrt[d_n]{d_n!}. \]
This bound has been reproved by \textcite{Sch78, Rad97} using entropy-based methods. In these works, some generalizations have been suggested for non-binary matrices $A\in \R_{\geq 0}^{n\times n}$, but as far as we know none imply \cref{thm:upperbound}. We remark that entropy-based bounds have recently found applications in counting bases of matroids and matroid intersections \cite{AOV18}; perfect matchings are a special case of matroid intersections.

Our method for proving \cref{thm:upperbound} is inspired by the entropy-based proofs of the Bregman-Minc inequality, but differs from them in key places. To illustrate, prior entropy-based proofs of Bregman-Minc proceeded by comparing, in terms of entropy-based measures, two distributions on perfect matchings in the bipartite graph associated with $A$. One was the uniform distribution $\mu$ over perfect matchings, and the other one $\nu$ can be described by the following sampling procedure: Iterate over vertices on one side, and each time pick a neighbor for the current vertex uniformly at random out of all currently unmatched neighbors. We remark that the latter distribution has been used to estimate the permanent by importance sampling, and an entropy-based measure of distance between $\mu$ and $\nu$ has been shown to have a close connection to the sample complexity \cite{Dia18, CD18}.

The generalization of $\mu$ to arbitrary $A\in \R_{\geq 0}^{n\times n}$ is straightforward. We just let $\mu(\sigma)\propto \prod_{i=1}^n A_{i,\sigma(i)}$ for $\sigma\in S_n$. Similarly there is a natural generalization of $\nu$: In iteration $i$ we just sample an unmatched neighbor $j$ with probability proportional to $A_{i,j}$. Indeed this generalization has been studied by \textcite{Sch78} in order to obtain generalizations of Bregman-Minc to non-binary matrices.

Our main insight is that defining $\nu$ according to another matrix $P\in \R_{\geq 0}^{n\times n}$ can result in a distribution that is closer to $\mu$. To be more precise, we find a matrix $P$ other than $A$, and modify the sampling procedure defining $\nu$ to pick a neighbor $j$ for $i$ with probability $\propto P_{i, j}$. We will ultimately let $P$ be the marginals of $\mu$, i.e., $P_{i, j}=\PrX{\sigma\sim \mu}{\sigma(i)=j}$. One feature of this modification for example is that the modified $\nu$ never picks any edge that can never be part of a sample of $\mu$. Comparing $\mu$ and the modified $\nu$ ultimately allows us to prove $\cref{thm:upperbound}$.

Our technique for upper bounding the permanent is completely analytical and is described in \cref{sec:upperbound}. However the approximation ratio guarantee resulting from it is only described in terms of the maximum of a well-defined universal function on the standard simplex. The maxima of this function can be easily guessed; but in order to prove that they are indeed the maxima, further ideas are needed. We combine several analytic ideas to reduce the possible set of maxima to a finite set (of moderate size) and with the help of a computer search formally prove that the guessed maxima of this function are indeed the global maxima, yielding the desired approximation ratio of $\sqrt{2}^n$. We remark that even though we used a computer search, we analytically show that the search can be done over a particular finite set of points and with exact rational arithmetic. So the computer program's successful execution should be considered a formal \emph{certificate} of correctness. For details of the computer search and the proof that the search can be reduced to a finite set see \cref{sec:phi}.

As a corollary of our upper bound technique, we prove \cref{thm:bp}, resolving a conjecture of \textcite{CY13}. Unlike the proof of our main result, where we use a computer search, the proof of \cref{thm:bp} involves simple analytic inequalities in addition to what is already proven in \cref{sec:upperbound}. The details are in \cref{sec:bp}. Readers interested in a proof of our main results \cref{thm:main,thm:upperbound}, can safely ignore this corollary and the entirety of \cref{sec:bp}.

\subsection{Structure of the paper}
In \cref{sec:prelim,sec:bethe}, we provide preliminaries on the Bethe approximation and recite examples which show the necessity of the approximation factor $\sqrt{2}^n$. \Cref{sec:upperbound} is the main technical section, where we derive an upper bound on the permanent based on an entropy-based measure of distance between distributions on perfect matchings. In \cref{sec:phi}, we massage this upper bound to prove our main result, \cref{thm:upperbound}. In \cref{sec:bp}, we massage the same upper bound to prove an alternative upper bound on the permanent, \cref{thm:bp}. \Cref{sec:phi,sec:bp} are independent of each other, and can be read while safely skipping the other.

%% file: prelim.tex
\section{Notation and preliminaries}\label{sec:prelim}

Everywhere $\log$ is taken in base $e$. We use $[n]$ to denote the set $\set{1,\dots,n}$. For a set $S$ we denote the family of all subsets of $S$ by $2^S$. We use the notation $\1$ to denote the vector of all ones in $\R^n$. For a set $S\subseteq [n]$, we use $\1_S$ to denote the indicator vector of $S$ which belongs to $\set{0, 1}^n$. We use $\Delta_n$ to denote the standard simplex\footnote{We remark that in geometry this set is usually denoted $\Delta_{n-1}$ but we deviate from the standard slightly and use $\Delta_n$.} of all $n$-dimensional stochastic vectors, i.e., all $p\in \R_{\geq 0}^n$ such that $\1^\intercal p=1$:
\[ \Delta_n=\set{(p_1,\dots,p_n)\in \R_{\geq 0}^n\given p_1+\dots+p_n=1}. \]

We call a matrix $P\in \R_{\geq 0}^{n\times n}$ doubly stochastic when each row and each column of $P$ belongs to $\Delta_n$, i.e., $P\1=\1$ and $\1^\intercal P=\1^\intercal$. The set of all doubly stochastic matrices is known as the Birkhoff polytope which we denote by $\B_n$:
\[ \B_n = \set{P\in \R_{\geq 0}^{n\times n}\given P\1=\1 \wedge \1^\intercal P=\1^\intercal}. \]
The fact that $\B_n$ is a polytope can be easily seen by noticing that it is defined by $2n$ linear equalities and $n^2$ linear inequalities. An alternative description of $\B_n$ is as the convex hull of permutation matrices \cite[see, e.g.,][]{Zie12}.

For a pair $e=(i, j)\in [n]\times [n]$, we use $A_e=A_{i, j}$ to denote the entry in the $i$-th row and $j$-th column of $A$. We use $S_n$ to denote the set of all permutations $\sigma:[n]\to [n]$. With some abuse of notation we also view $\sigma$ as the set $\set{(i, \sigma(i))\given i\in [n]}\subseteq [n]\times [n]$. So in particular $e\in \sigma$ means that $e$ is a pair of the form $(i, \sigma(i))$ for some $i\in [n]$. With this notation the permanent can be written as
\[ \per(A)=\adjustlimits\sum_{\sigma\in S_n} \prod_{e\in \sigma} A_e. \]

We sometimes think of $\sigma\in S_n$ as a perfect matching in a bipartite graph: The first side is the domain of $\sigma$, and the other side the image of $\sigma$. For two permutations $\sigma, \pi\in S_n$ we take $\sigma\circ \pi$ to be the composed permutation defined by $\sigma\circ\pi(i)=\sigma(\pi(i))$. We also sometimes view a permutation $\pi\in S_n$ as a total ordering on $\set{1,\dots,n}$, that is we define the ordering $<_\pi$ on $[n]$ by $\pi(1)<_\pi \pi(2)<_\pi\dots<_\pi\pi(n)$. For a permutation $\pi$, we let $\overline{\pi}$ denote the reverted permutation defined by $\overline{\pi}(i)=\pi(n-i+1)$. The total ordering defined by $\pi$ and $\overline{\pi}$ are the opposites of each other.

For two distributions $\mu$ and $\nu$ supported on a finite set $\Omega$, we define their Kullback-Leibler divergence, KL divergence for short, as
\[ \KL{\mu\vs \nu}=\sum_{\omega\in \Omega} \mu(\omega)\log\frac{\mu(\omega)}{\nu(\omega)}=\ExX*{\omega\sim \mu}{\log \frac{\mu(\omega)}{\nu(\omega)}}. \]
An important well-known property of the KL divergence is nonnegativity. This can be proved for example by applying Jensen's inequality to the convex function $\log(1/\cdot)$ \cite[see, e.g.,][]{Mac03}.
\begin{fact}[Gibbs' inequality]\label{fact:KL}
	For any two distributions $\mu$ and $\nu$ we have $\KL{\mu\vs \nu}\geq 0$.
\end{fact}

%% file: bethe.tex
\section{Bethe permanent}\label{sec:bethe}

Doubly stochastic matrices have an intimate connection to random permutations. They can be thought of as marginals of distributions on permutations. More precisely, for any distribution $\mu$ on $S_n$, the $n\times n$ matrix $P$ defined by $P_{i, j}=\PrX{\sigma\sim \mu}{(i, j)\in \sigma}$ is doubly stochastic. Alternatively for any $P\in \B_n$, there is $\mu$ supported on $S_n$ whose marginals are $P$; this is because the vertices of $\B_n$ are permutation matrices.

Given a matrix $A\in \R_{\geq 0}^{n\times n}$, a natural distribution $\mu$ on $S_n$ can be defined by $\mu(\sigma)\propto \prod_{e\in \sigma} A_e$. The partition function, or the normalizing constant in this definition is $\per(A)$. Given these connections, it is no surprise that many upper and lower bounds on the $\per(A)$ can be expressed in terms of doubly stochastic matrices. For example the upper bound implicitly used by \textcite{LSW00} is
\[ \log \per(A)\leq \max\set*{\sum_{e\in [n]\times [n]} P_e \log(A_e/P_e)\given P\in \B_n}, \]
which can be interpreted as a form of subadditivity of entropy. The Van der Waerden conjecture used as a lower bound by \textcite{LSW00} can be rewritten as
\[ n+\log\per(A)\geq \sum_{e\in [n]\times [n]} P_e\log(A_e/P_e), \]
for all $P\in \B_n$. For a more detailed account of these connections see \textcite{Gur11}. We remark that the $P$ maximizing the r.h.s.\ of the above can also be obtained by matrix scaling. In other words, as long as $\per(A)>0$, there are sequences of diagonal matrices $R_1,R_2,\dots$ and $C_1,C_2,\dots$ such that
\[ \lim_{i\to \infty} R_iAC_i = \argmax\set*{\sum_{e\in [n]\times [n]}P_e\log(A_e/P_e)\given P\in \B_n}. \]

\textcite{CKV08, HJ09} observed that the $\mu$ distribution defined for a matrix $A$ can be described by a graphical model. Roughly speaking, one can define factors $R_1,\dots,R_n: 2^{[n]\times [n]}\to \R_{\geq 0}$ and $C_1,\dots,C_n: 2^{[n]\times [n]}\to\R_{\geq 0}$ where for $\sigma\subseteq [n]\times[n]$
\[ R_i(\sigma)=\begin{cases}
	\sqrt{A_{i, j}} & \text{if }\card{\sigma\cap (\set{i}\times [n])}=1\text{ and }{\sigma\cap (\set{i}\times [n])=\set{(i, j)}},\\ 
	0 & \text{if }\card{\sigma\cap (\set{i}\times [n])}\neq 1,
\end{cases}\]
\[ C_j(\sigma)=\begin{cases}
	\sqrt{A_{i, j}} & \text{if }\card{\sigma\cap ([n]\times \set{j})}=1\text{ and }{\sigma\cap ([n]\times \set{j})=\set{(i, j)}},\\ 
	0 & \text{if }\card{\sigma\cap ([n]\times \set{j})}\neq 1.
\end{cases}\]
With this definition we have
\[ \parens*{\prod_{i} R_i(\sigma)}\parens*{\prod_{j} C_j(\sigma)}=\begin{cases}
	\prod_{e\in \sigma} A_e & \text{if }\sigma\text{ is a permutation},\\
	0& \text{otherwise}.
\end{cases}\]
A practically successful heuristic for computing the partition function of graphical models is Bethe approximation; in fact for ``loopless'' graphical models the Bethe approximation exactly computes the partition function. The graphical model defined here is not ``loopless'', but the Bethe approximation can still be defined and in fact computed efficiently. For a detailed account of these, see \textcite{Von13}. We only describe the Bethe approximation for the aforementioned graphical model here, but we remark that Bethe approximation can be defined for arbitrary graphical models.
\begin{definition}\label{def:beta}
	For a matrix $A\in \R_{\geq 0}^{n\times n}$ and a doubly stochastic matrix $P\in \B_n$, define
	\[ \beta(A, P)=\sum_{e\in [n]\times [n]}\parens*{P_e \log(A_e/P_e)+(1-P_e)\log(1-P_e)}. \]	
\end{definition}
Note that according to this definition and \cref{def:BP} we have
\[ \bethe(A)=\max\set{e^{\beta(A, P)}\given P\in \B_n}. \]
\Textcite{Von13} proved that $\beta(A, P)$ is a concave function of $P$. \Textcite{Gur11} proved that $\beta(A,P)$ is always a lower bound for $\log(\per(A))$ by using an inequality of \textcite{Sch98}. For two recent alternative proofs see \textcite{Csi14, AO17}.

The matrix $P$ maximizing $\beta(A, P)$ is not necessarily obtained from $A$ by matrix scaling. It can be either obtained by convex programming methods since $\beta(A, P)$ is concave in $P$, or by the sum-product algorithm \cite{Von13}. Despite this, \textcite{GS14} analyzed the objective $\beta(A, P)$ for the matrix $P$ obtained from matrix scaling; in other words they analyzed $\beta(A, RAC)$ where $R, C$ are diagonal matrices such that $RAC\in \B_n$. They showed that
\[ \per(A)\leq 2^n e^{\beta(A, RAC)}. \] 

We choose a different doubly stochastic matrix. We analyze $\beta(A, P)$ where $P$ is the doubly stochastic matrix of $\mu$'s marginals. In other words we let $P_{i, j}=\PrX{\sigma\sim \mu}{(i, j)\in \sigma}$, where $\mu(\sigma)\propto \prod_{e\in \sigma} A_e$.
\begin{theorem}\label{thm:marginals}
	For $P$ defined as above, we have
	\[ \per(A)\leq \sqrt{2}^n e^{\beta(A, P)}. \]
\end{theorem}
Note that \cref{thm:marginals} immediately implies \cref{thm:upperbound}. We will prove \cref{thm:marginals} in \cref{sec:upperbound}, but in the rest of this section we demonstrate why the factor $\sqrt{2}^n$ is tight. The example we use was already studied in \cite{GS14}, but for completeness we mention it here as well.

First let us compute $\beta(A, P)$, where $A$ is a $2\times 2$ matrix. Note that $\B_2$ is a one-dimensional family that can be described by a single parameter $p\in [0,1]$:
\[ \B_2=\set*{\begin{bmatrix}
	p& 1-p\\
	1-p& p\\
\end{bmatrix}\given p\in [0, 1]}. \]
Then it is an easy exercise to see $\beta(A, P)=p\log(A_{1, 1}A_{2, 2})+(1-p)\log(A_{1, 2}A_{2, 1})$, and as such $\bethe(A)=\max\set{A_{1,1}A_{2,2}, A_{1, 2}A_{2,1}}$. Clearly there can be a factor of $2$ difference between $\max\set{A_{1,1}A_{2,2}, A_{1, 2}A_{2, 1}}$ and the sum, $\per(A)=A_{1,1}A_{2,2}+A_{1, 2}A_{2, 1}$. This happens for example for the matrix
\[ A=\begin{bmatrix}
	1 & 1\\
	1 & 1\\
\end{bmatrix}, \]
for which $\per(A)=2\bethe(A)=\sqrt{2}^2\bethe(A)$. This shows that $\sqrt{2}^n$ is tight for $n=2$. 

For higher values of $n$, note that both $\per(\cdot)$ and $\bethe(\cdot)$ behave predictably w.r.t.\ block-diagonal matrices. In other words if $B\in \R_{\geq 0}^{n_1\times n_1}$ and $C\in \R_{\geq 0}^{n_2\times n_2}$, and $A$ is defined as
\[ A=\begin{bmatrix}
	B & 0 \\
	0 & C\\
\end{bmatrix}, \]
then $\per(A)=\per(B)\per(C)$ and $\bethe(A)=\bethe(B)\bethe(C)$. To see the latter, note that if $P\in \B_{n_1+n_2}$ has any nonzero entry in places where $A$ has a zero entry, then $\beta(A, P)=-\infty$. So to compute $\bethe(A)$, it is enough to maximize $\beta(A, P)$ over doubly stochastic matrices $P$ which also have a block-diagonal structure
\[ P=\begin{bmatrix}
	Q & 0\\
	0 & R\\
\end{bmatrix}. \]
But $P\in \B_n$ iff $Q\in \B_{n_1}$ and $R\in \B_{n_2}$. This shows that $\beta(A, P)=\beta(B, Q)+\beta(C, R)$ and consequently $\bethe(A)=\bethe(B)\bethe(C)$.
So starting from any matrix, we can construct higher-dimensional matrices by using a block-diagonal form and amplify the gap between $\per(\cdot)$ and $\bethe(\cdot)$. In particular for the $n\times n$ matrix
\[ A=I_{n/2}\otimes \begin{bmatrix}
	1 & 1\\
	1 & 1\\
\end{bmatrix}=\begin{bmatrix}
	1 & 1 & 0 & 0 & \dots & 0\\
	1 & 1 & 0 & 0 & \dots & 0\\
	0 & 0 & 1 & 1 & \dots & 0\\
	0 & 0 & 1 & 1 & \dots & 0\\
	\vdots & \vdots & \vdots & \vdots &\ddots & \vdots\\
	0 & 0 & 0 & 0 & \dots & 1\\
\end{bmatrix}, \]
we have $\per(A)=2^{n/2}=\sqrt{2}^n$ and $\bethe(A)=1$. Finally, we remark that it is also easy to construct examples $A$ where $\per(A)=\bethe(A)$. For example this happens when $A$ is a diagonal matrix, such as the identity matrix.

%% file: upperbound.tex
\section{Entropy-based upper bound on the permanent}\label{sec:upperbound}

In this section we prove \cref{thm:marginals}. Our strategy is to upper bound the permanent by comparing two distributions on perfect matchings in terms of the KL divergence. We will ultimately show that the gap in our upper bound is related to the maximum of a function $\phi:\Delta_n\to \R$ which is independent of the matrix $A$. Then, in \cref{sec:phi} we analyze this function and obtain its maximum. We find it instructive to have the tight example from the end of \cref{sec:bethe} in mind throughout the proof and notice that the inequalities used are tight for this example.

Let $A\in \R_{\geq 0}^{n\times n}$ be a fixed matrix. Since we are upper bounding the permanent by a nonnegative quantity, we can assume w.l.o.g.\ that $\per(A)>0$. Then we can define a distribution $\mu$ on $S_n$ by $\mu(\sigma)\propto \prod_{e\in \sigma} A_e$, or in other words
\[ \mu(\sigma)=\frac{\prod_{e\in \sigma} A_e}{\per(A)}. \]
We will also fix a doubly stochastic matrix $P$ which is defined by
\[ P_{i, j}=\PrX{\sigma\sim \mu}{(i, j)\in \sigma}. \]
$P$ is doubly stochastic because for each $\sigma\in S_n$ and every $i$ there is a \emph{unique} $j$ such that $(i, j)\in \sigma$, and for every $j$ there is a \emph{unique} $i$ such that $(i, j)\in \sigma$.

Now let us fix an ordering $\pi\in S_n$ over rows of $P$ and define another distribution $\nu_\pi$ on $S_n$ by the following sampling procedure: We iterate over rows of $P$ according to the order $\pi$. In iteration $i$ we select a yet unselected column $j$ with probability proportional to $P_{\pi(i), j}$. We let $\sigma$ be the set of all $(\pi(i), j)$ selected by this procedure.

Let us compute $\nu_\pi(\sigma)$ explicitly for some given $\sigma$. When the sampling procedure produces $\sigma\in S_n$, in iteration $i$, the unselected columns are exactly $\sigma\circ \pi(i), \sigma\circ\pi(i+1),\dots,\sigma\circ\pi(n)$. So we pick the first one, $\sigma\circ\pi(i)$, with probability
\[ \frac{P_{\pi(i), \sigma\circ\pi(i)}}{\sum_{j=i}^n P_{\pi(i), \sigma\circ\pi(j)}}. \]
This means that
\begin{equation}\label{eq:prob}
	\nu_\pi(\sigma)=\prod_{i=1}^n\frac{P_{\pi(i), \sigma\circ\pi(i)}}{\sum_{j=i}^n P_{\pi(i), \sigma\circ\pi(j)}}=\prod_{i=1}^n \frac{P_{i, \sigma(i)}}{\sum_{j\geq_{\sigma\circ\pi} \sigma(i)} P_{i, j}},
\end{equation} 
where as a reminder $j\geq_{\sigma\circ\pi} \sigma(i)$ means that $j=\sigma(i)$ or $j$ appears after $\sigma(i)$ in the list $\sigma\circ\pi(1),\dots,\sigma\circ\pi(n)$.

It is worth mentioning that if $A$ was the tight example presented at the end of \cref{sec:bethe}, regardless of the choice of $\pi$, we would have $\mu=\nu_\pi$. But this is generally not true and $\mu$ can be quite different from $\nu_\pi$. Regardless, by \cref{fact:KL}, we have
\begin{equation}\label{eq:KL}
	\KL{\mu\vs \nu_\pi}=\ExX*{\sigma\sim \mu}{\log\frac{\mu(\sigma)}{\nu_\pi(\sigma)}}\geq 0.
\end{equation}
Now we decompose \cref{eq:KL} part by part. First we have
\begin{align}\label{eq:numerator}
	\ExX{\sigma\sim \mu}{\log \mu(\sigma)}&=\ExX*{\sigma\sim \mu}{\log \frac{\prod_{e\in \sigma} A_e}{\per(A)}}=\sum_{e\in [n]\times[n]}\PrX{\sigma\sim \mu}{e\in \sigma}\log(A_e)-\log(\per(A))\nonumber\\
	&=\sum_{e\in [n]\times[n]} P_e\log(A_e)-\log(\per(A)),
\end{align}
where the last equality follows from our choice of $P$, for which we have $\PrX{\sigma\sim \mu}{e\in \sigma}=P_e$. For the second part we use \cref{eq:prob} to derive
\begin{equation}\label{eq:denom}
	\ExX{\sigma\sim \mu}{\log \nu_\pi(\sigma)}=\ExX*{\sigma\sim \mu}{\log \parens*{\prod_{i=1}^n \frac{P_{i, \sigma(i)}}{\sum_{j\geq_{\sigma\circ\pi} \sigma(i)} P_{i, j}}}}=\sum_{e\in [n]\times[n]} P_e\log(P_e)-\sum_{i\in [n]}\ExX*{\sigma\sim \mu}{\log\sum_{j\geq_{\sigma\circ\pi} \sigma(i)} P_{i, j}}
\end{equation}

Combining \cref{eq:KL,eq:numerator,eq:denom} we get the inequality
\begin{equation}\label{eq:mainineq}
	\sum_{e\in [n]\times [n]}P_e\log(A_e/P_e)+\sum_{i\in [n]}\ExX*{\sigma\sim \mu}{\log\sum_{j\geq_{\sigma\circ\pi} \sigma(i)} P_{i,j}}\geq \log(\per(A)).
\end{equation}
This already looks very promising as we have $\log(\per(A))$ on one side and part of $\beta(A, P)$ on the other -- see \cref{def:beta}. However the extra sum does not look easy to analyze. In order to simplify it, we take the expectation of \cref{eq:mainineq} over $\pi$ chosen \emph{uniformly} at random from $S_n$, which we denote by $\pi\sim S_n$. Then we get
\begin{equation*}
	\sum_{e\in [n]\times [n]}P_e\log(A_e/P_e)+\sum_{i\in [n]}\ExX*{\sigma\sim \mu}{\ExX*{\pi\sim S_n}{\log\sum_{j\geq_{\sigma\circ\pi} \sigma(i)} P_{i, j}}}\geq \log(\per(A)).
\end{equation*}
The main benefit here is that for any $\sigma$, when $\pi$ is a uniformly random element of $S_n$, $\sigma\circ \pi$ is also a uniformly random element of $S_n$. So we can simplify this to
\begin{equation*}
	\sum_{e\in [n]\times [n]}P_e\log(A_e/P_e)+\sum_{i\in [n]}\ExX*{\sigma\sim \mu}{\ExX*{\pi\sim S_n}{\log\sum_{j\geq_{\pi} \sigma(i)} P_{i, j}}}\geq \log(\per(A)).
\end{equation*}
Notice that we have removed all dependence on $\sigma$ except for $\sigma(i)$. So we can simplify and get
\begin{equation}\label{eq:entropyub}
	\sum_{e\in [n]\times [n]}P_e\log(A_e/P_e)+\sum_{i\in [n]}\sum_{k\in [n]}P_{i, k}\cdot \ExX*{\pi\sim S_n}{\log\sum_{j\geq_{\pi} k} P_{i, j}}\geq \log(\per(A)).
\end{equation}
Now all appearances of $\mu$ have been removed. Our goal is to show that $\beta(A, P)+n\log(\sqrt{2})\geq \log(\per(A))$. Comparing with the above and \cref{def:beta}, it is enough to prove the following for each $i=1,\dots,n$:
\[ \log(\sqrt{2})+\sum_{k\in [n]}(1-P_{i, k})\log(1-P_{i, k})\geq \sum_{k\in [n]} P_{i, k}\cdot \ExX*{\pi\sim S_n}{\log \sum_{j\geq_\pi k} P_{i, j}}. \]
Proving this would finish the proof of \cref{thm:marginals}. Note that this is a statement that only involves the $i$-th row of $P$. So we can abstract away and prove the following lemma.
\begin{lemma}\label{lem:perms}
	Suppose that $(p_1,\dots,p_n)\in \Delta_n$. Then
	\[ \log(\sqrt{2})\geq \ExX*{\pi\sim S_n}{\sum_{k=1}^n\parens*{p_k \log\parens*{\sum_{j\geq_\pi k}p_j}-(1-p_k)\log(1-p_k)}}.  \]
\end{lemma}
\begin{proof}
	Ideally we would like to prove the inequality for each $\pi\in S_n$ and remove the expectation over $\pi$.
	 However, the inequality does not hold individually for each $\pi\in S_n$. Instead we pair up permutations $\pi$ with their reverses $\overline{\pi}$, and prove that the sum of the value inside the r.h.s.\ expectation over $\pi$ and $\overline{\pi}$ for any fixed $\pi$ is at most twice the l.h.s. This would finish the proof. Note that the only appearance of $\pi$ on the r.h.s.\ is in the ordering $\geq_\pi$ it induces. So up to reindexing of $p_1,\dots,p_n$ it is enough to prove this for the identity permutation, i.e., it is enough to prove
	\[ \log(2)\geq \sum_{k\in [n]} p_k\log(p_1+\dots+p_k)+\sum_{k\in [n]}p_k\log(p_k+\dots+p_n)-2\sum_{k\in [n]}(1-p_k)\log(1-p_k). \]
	Naming the r.h.s.\ of the above as $\phi(p_1,\dots,p_n)$, we will prove in \cref{sec:phi}, in \cref{lem:ps}, that $\max\set{\phi(p)\given p\in \Delta_n}$ is at most $\log(2)$. This finishes the proof since
	\[ \ExX*{\pi\sim S_n}{\sum_{k=1}^n\parens*{p_k \log\parens*{\sum_{j\geq_\pi k}p_j}-(1-p_k)\log(1-p_k)}}=\frac{1}{2}\ExX*{\pi\sim S_n}{\phi(p_{\pi(1)},\dots,p_{\pi(n)})}\leq \log(\sqrt{2}). \]
\end{proof}

%% file: phi.tex
\section{Analyzing the maximum gap}\label{sec:phi}
In this section we prove that the maximum of the function $\phi:\Delta_n\to \R$ is $\log(2)$, finishing the proof of \cref{thm:marginals}. As a reminder, we define the function $\phi$ again below.
\begin{definition}\label{def:pk}
	For $p=(p_1,\dots,p_n)\in \Delta_n$ define
	\[ \phi(p_1,\dots,p_n)=\sum_{k\in [n]}p_k\log(p_1+\dots+p_k)+\sum_{k\in [n]}p_k\log(p_k+\dots+p_n)-2\sum_{k\in [n]}(1-p_k)\log(1-p_k). \]
\end{definition}
With some abuse of notation we use the same letter $\phi$ to refer to this function over $\Delta_n$ for any $n$.

A simple calculation shows that
\[ \phi(1/2,1/2)=\log(2). \]
Also notice that zero coordinates in the input to $\phi$ can be dropped without changing the value of $\phi$. In other words, for any $p\in \Delta_{n-1}$ and any $i$, we have
\[ \phi(p_1,\dots,p_i,0,p_{i+1},\dots,p_{n-1})=\phi(p_1,\dots,p_{n-1}). \]
This shows that for any two distinct elements $\set{i, j}\subseteq [n]$ we have
\[ \phi(\1_{\set{i, j}}/2)=\log(2). \]
Our goal is to show that these $\binom{n}{2}$ points are the global maxima of $\phi$ over $\Delta_n$. The main result of this section is
\begin{lemma}\label{lem:ps}
	For any $p\in \Delta_n$ we have $\phi(p)\leq \log(2)$.
\end{lemma}
Our strategy for proving \cref{lem:ps} is to prove it for $n=1,2,3$ and then reduce general $n$ to one of these cases. The main tool enabling this reduction is the following lemma.
\begin{lemma}\label{lem:reduction}
	Suppose that $(q,r,s,t)\in \Delta_4$ and $r+s\leq \gamma$ where $\gamma=(\sqrt{17}-3)/2\simeq 0.56$ is a constant. Then
	\[ \phi(q,r,s,t)\leq \phi(q,r+s,t). \]
\end{lemma}
Let us now see why \cref{lem:ps} can be proved, assuming the cases $n=1,2,3$ and \cref{lem:reduction}.
\begin{proof}[Proof of \cref{lem:ps} assuming the proof for $n=1,2,3$ and \cref{lem:reduction}]
	Suppose that $p=(p_1,\dots,p_n)\in \Delta_n$. If there is an $i$ such that $p_i+p_{i+1}\leq \gamma$, then we claim that
	\[ \phi(p_1,\dots,p_n)\leq \phi(p_1,\dots,p_{i-1},p_i+p_{i+1},p_{i+2},\dots,p_n). \]
	To see this let $q=p_1+\dots+p_{i-1}$, $r=p_i$, $s=p_{i+1}$, and $t=p_{i+2}+\dots+p_n$. Notice that $(q,r,s,t)\in \Delta_4$ and it follows from \cref{def:pk} that
	\[ \phi(p_1,\dots,p_{i-1},p_i+p_{i+1},p_{i+2},\dots,p_n)-\phi(p_1,\dots,p_n)=\phi(q,r+s,t)-\phi(q,r,s,t), \]
	and by \cref{lem:reduction} this quantity is nonnegative.
	
	We just proved that if there is a consecutive pair $p_i, p_{i+1}$ in $p_1,\dots,p_n$ that add to at most $\gamma$ we can reduce $n$ without decreasing $\phi$, by replacing $p_i, p_{i+1}$ with $p_i+p_{i+1}$. Notice that such a pair always exists for $n\geq 4$. This is because $(p_1+p_2)+(p_3+p_4)\leq 1$, so at least one of $p_1+p_2$ and $p_3+p_4$ must be $\leq 1/2<\gamma$. So by induction the problem gets reduced to the cases $n=1,2,3$.
\end{proof}
Next we prove \cref{lem:reduction}. Afterwards we prove $\cref{lem:ps}$ for $n=1,2,3$ finishing the whole proof.
\begin{proof}[Proof of \cref{lem:reduction}]
	We want to show that $\phi(q,r,s,t)-\phi(q,r+s,t)\leq 0$. To reduce the number of variables, let us fix $r,s$ and choose $q, t$ in a way that $\phi(q,r,s,t)-\phi(q,r+s,t)$ is maximized. It is then enough to show that the maximum is nonpositive. One can compute that
	\begin{multline*}
		\phi(q,r,s,t)-\phi(q,r+s,t)=r\log\parens*{1-\frac{s}{q+r+s}}+s\log\parens*{1-\frac{r}{r+s+t}}\\
		-2(1-r)\log(1-r)-2(1-s)\log(1-s)+2(1-r-s)\log(1-r-s).
	\end{multline*}
	If we replace $q, t$ by $q+\epsilon, t-\epsilon$ for some $\epsilon\in [-q, t]$, then only the first two terms in the above change. One can easily check that the above is a concave function of $\epsilon$ by checking concavity of the first two summands in terms of $q, t$. So in order to find the $\epsilon$ maximizing the above, we can set the derivative w.r.t.\ $\epsilon$ to zero. If $q, t$ were already the maximizers, then either the derivative at $\epsilon=0$ would vanish or one of $q, t$ would be $0$. We will show that the first case happens for $r+s\leq \gamma$, and obtain an explicit expression for the maximizing $q,t$. Let us compute the derivative:
	\[ \eval{\frac{d}{d\epsilon}\parens*{\phi(q+\epsilon,r,s,t-\epsilon)-\phi(q+\epsilon,r+s,t-\epsilon)}}_{\epsilon=0}=\frac{s}{(q+r)(q+r+s)}-\frac{r}{(s+t)(r+s+t)}. \]
	One can compute that this derivative vanishes at
	\begin{equation}\label{eq:optqt}
		q^*(r,s)=\frac{1-r(1+r+s)}{2+r+s}\qquad\text{and}\qquad t^*(r,s)=\frac{1-s(1+r+s)}{2+r+s},
	\end{equation}
	and for this choice of $q=q^*(r,s), t=t^*(r,s)$ we have $q+r+s+t=1$. So as long as $q^*(r,s),t^*(r,s)\geq 0$, this will be the maximizer of $\phi(q,r,s,t)-\phi(q,r+s,t)$ for fixed $r,s$. But $q^*(r,s),t^*(r,s)\geq 0$ exactly when $(1+r+s)\cdot \max(r,s)\leq 1$. Since we have assumed $s+t\leq \gamma$, we can deduce
	\[ (1+r+s)\max(r,s)\leq (1+\gamma)\gamma<1, \]
	where the last inequality can be easily checked for the constant $\gamma=(\sqrt{17}-3)/2$. So from now on, we assume that \cref{eq:optqt} holds and $q=q^*(r,s)$ and $t=t^*(r,s)$. With this choice we can define
	\begin{multline*}
		\psi(r,s):=\phi(q^*(r,s),r,s,t^*(r,s))-\phi(q^*(r,s),r+s,t^*(r,s))=\\
		-(r+s)\log(1+r+s)+(s-r)\log\parens*{\frac{1+s}{1+r}}\\
		-2(1-r)\log(1-r)-2(1-s)\log(1-s)+2(1-r-s)\log(1-r-s).
	\end{multline*}
	Our goal is to prove that $\psi(r,s)\leq 0$ as long as $r+s\leq \gamma$. Let us call $C=r+s$. We can look at the function $\psi(x, C-x)$ as $x$ varies from $0$ to $C$; then our goal would be to show $x\mapsto \psi(x, C-x)$ is nonpositive over this range, as long as $C\leq \gamma$. Note that $\psi(\cdot, \cdot)$ is symmetric in its inputs, so it is enough to prove $x\mapsto \psi(x, C-x)$ is nonpositive for $x\in [0, C/2]$. For this we will show that for $x\in [0, C/2]$ we have $\frac{d}{dx}\psi(x, C-x)\leq 0$. This would finish the proof because assuming the derivative is nonpositive, we have $\psi(0, C)=0$ and
	\[ \psi(x, C-x)=\int_{0}^x \frac{d}{dy}\psi(y, C-y)dy\leq 0. \]
	So all that remains to finish the proof is to show that $\frac{d}{dx}\psi(x, C-x)\leq 0$ when $x\in [0, C/2]$ and $C\leq \gamma$. For brevity of expressions we use $r=x$ and $s=C-x$ and notice that $\frac{d}{dx}r=1$ and $\frac{d}{dx}s=-1$. The term $r+s$ and functions of $r+s$ are constant w.r.t.\ $x$, so their derivatives vanish. This simplifies the calculations and we obtain $\frac{d}{dx}\psi(r, s)=\frac{d}{dx}\psi(x, C-x)=$
	\begin{multline*}
		\frac{d}{dx}\parens*{(s-r)\log\parens*{\frac{1+s}{1+r}}-2(1-r)\log(1-r)-2(1-s)\log(1-s)}=\\
		-2\log\parens*{\frac{1+s}{1+r}}-(s-r)\parens*{\frac{1}{1+r}+\frac{1}{1+s}}+2\log\parens*{\frac{1-r}{1-s}}=\\
		(s-r)\cdot\frac{\eval{2\log\parens*{\frac{1}{1-z^2}}}_{z=r}^{z=s}}{s-r}-(s-r)\cdot \parens*{\frac{1}{1+r}+\frac{1}{1+s}}.
	\end{multline*}
	Noting that $s-r\geq 0$, it is enough to prove that
	\[ \frac{\eval{2\log\parens*{\frac{1}{1-z^2}}}_{z=r}^{z=s}}{s-r}\leq \frac{1}{1+r}+\frac{1}{1+s}. \]
	Note that $\frac{d}{dz}(2\log(1/(1-z^2)))=4z/(1-z^2)$. So we can rewrite the above as
	\[ \frac{\int_{r}^s \frac{4z}{1-z^2}dz}{s-r}\leq \frac{1}{1+r}+\frac{1}{1+s}. \]
	The l.h.s.\ is the average value of $4z/(1-z^2)$ over the interval $[r, s]$. The function $z\mapsto 4z/(1-z^2)$ is convex over $[r,s]\subseteq [0,1]$, so the average over $[r, s]$ can be upper bounded by the average over the two endpoints $\set{r, s}$. So it would be enough to prove
	\[ \frac{2r}{1-r^2}+\frac{2s}{1-s^2}\leq \frac{1}{1+r}+\frac{1}{1+s}. \]
	Multiplying all sides by $(1-r^2)(1-s^2)$ and rearranging the terms it would be enough to prove the following
	\[ (1-3r)(1-s^2)+(1-3s)(1-r^2)\geq 0. \]
	The l.h.s.\ can be rewritten as
	\[ 2-3(r+s)(1-rs)-(r^2+s^2)\geq 2-3(r+s)-(r+s)^2=2-3C-C^2, \]
	where for the inequality we used the fact that $1-rs\leq 1$ and $r^2+s^2\leq (r+s)^2$. The above is a quadratic in $C$ which has roots at $\frac{1}{2}(\pm\sqrt{17}-3)$. So it is nonnegative between the roots, and in particular on the interval $C\in [0, (\sqrt{17}-3)/2]=[0,\gamma]$.
\end{proof}
Now we just need to prove \cref{lem:ps} for $n=1, 2, 3$.
\begin{proof}[Proof of \cref{lem:ps} for $n=1$]
	For $n=1$, the only $p\in \Delta_1$ is $p=1$. We can easily compute $\phi(1)=0\leq \log(2)$ as claimed.
\end{proof}
\begin{proof}[Proof of \cref{lem:ps} for $n=2$]
	We have $\Delta_2=\set{(q, 1-q)\given q\in[0,1]}$. So we can calculate and see that
	\[ \phi(q, 1-q)=-q\log(q)-(1-q)\log(1-q). \]
	This is a concave function of $q$ and it is symmetric on $[0, 1]$. So its maximum is achieved at $q=1/2$, where it takes the value $\phi(1/2,1/2)=\log(2)$.
\end{proof}
\begin{proof}[Proof of \cref{lem:ps} for $n=3$]
	We have to show that if $q,r,s\geq 0$ and $q+r+s=1$, then $\phi(q,r,s)\leq \log(2)$. This looks difficult to prove at first; there are at least three points where equality can be attained, namely $(1/2,0,1/2),(1/2,1/2,0),(0,1/2,1/2)$, and we want to prove these are the maxima. This suggests that convexity arguments would not work, and we need to somehow break the symmetry. One can gain some confidence about the correctness of the statement by using a computer program to calculate $\phi$ over an $\epsilon$-net of $\Delta_3$, but this would not constitute a proof. Instead we use some tricks to reduce the search space to a finite set and use exact rational arithmetic, in order to formally verify the statement with the help of a computer program.
	
	Our insight is that we only need to upper bound $\phi(q,r,s)$ for $(q,r,s)\in U\subseteq \Delta_3$, where $U$ is a subset of $\Delta_3$ which is far from the three points conjectured to be the maxima. To see this, note that if $q+r\leq \gamma$ or $r+s\leq \gamma$, then we can again use \cref{lem:reduction} to reduce to the case $n=2$ which we have already proved. So we can define
	\[ U=\set{(q,r,s)\in \Delta_3\given q+r\geq \gamma \wedge r+s\geq \gamma}=\set{(q,r,s)\in \Delta_3\given q,s\leq 1-\gamma}. \]
	Notice that none of the three conjectured maxima are inside $U$ as for each one of them either $q=1/2>1-\gamma$ or $s=1/2>1-\gamma$.
	
	So now we can potentially form an $\epsilon$-net for $U$, bound the error of approximating $\phi$ at any point by the closest net point and hope that the approximation error is smaller than $\log(2)-\max\set{\phi(p)\given p\in U}$. The latter would be a positive number, so taking $\epsilon$ small enough should in principle yield a proof.
	
	\begin{algorithm}
	\caption{An algorithm to verify $\max\set{\phi(q,r,s)\given (q,r,s)\in U}\leq \log(2)$.}
		$\epsilon\leftarrow 1/N$
		
		$M\leftarrow \lfloor 44N/100\rfloor$ \tcc{This is the upper bound for $i, j$, because $1-\gamma<0.44$.}
		
		\For{$i=0,1,\dots,M$}{
			\For{$j=0,1,\dots,M$}{
				$\qbar \leftarrow \frac{i}{N}$
				
				$\sbar \leftarrow \frac{j}{N}$
				
				\eIf{$i=j=0$ and $N>100$}{
					Verify the following:
					\[ (\qbar+\epsilon)^i (\sbar+\epsilon)^j\leq 2^N (1-\qbar-\epsilon)^{N-i+j+1}(1-\sbar-\epsilon)^{N+i-j+1}(\qbar+\sbar+2\epsilon)^{2(i+j+2)} \]
				}{
					Verify the following:
					\[ (\qbar+\epsilon)^i (\sbar+\epsilon)^j\leq 2^N (1-\qbar-\epsilon)^{N-i+j+1}(1-\sbar-\epsilon)^{N+i-j+1}(\qbar+\sbar)^{2(i+j+2)} \]
				}
			}
		}
	\label{alg:verify}
	\end{algorithm}
	\DeclareFixedFont{\ttb}{T1}{txtt}{bx}{n}{12} 
	\DeclareFixedFont{\ttm}{T1}{txtt}{m}{n}{12}  
	\begin{algorithm}
		\caption{Implementation of \cref{alg:verify} for $N=2000$ in Python 3.}
		\lstset{
			language=Python,
			basicstyle=\ttm,
			otherkeywords={self, assert},             
			keywordstyle=\ttb\color{Green},
			emph={Fraction},          
			emphstyle=\ttb\color{Orange},    
			stringstyle=\color{Green},
			showstringspaces=false            %
		}
\begin{lstlisting}[language=Python]
from fractions import Fraction

N = 2000
M = 880 # .44 * N
eps = Fraction(1, N)

for i in range(M+1):
    for j in range(M+1):
        q = Fraction(i, N)
        s = Fraction(j, N)
        LHS = (q+eps)**i * (s+eps)**j
        RHS = 2**N * (1-q-eps)**(N-i+j+1) * (1-s-eps)**(N+i-j+1) *\
                (q+s+(2*eps if i==0 and j==0 else 0))**(2*(i+j+2))
        assert LHS<=RHS
\end{lstlisting}
		\label{alg:python}
	\end{algorithm}
	
	The main problem with this approach would be computing $\phi$ over the $\epsilon$-net since we would potentially be dealing with irrational numbers that we cannot even write down. We can however resolve this by taking $\exp$ from both sides of the desired inequality $\phi(q,r,s)\leq \log(2)$. Let us first simplify the expression for $\phi(q,r,s)$ by replacing $r$ with $1-q-s$ and calculating
	\[ \phi(q,r,s)=q\log(q)+s\log(s)-(1-q+s)\log(1-q)-(1+q-s)\log(1-s)-2(q+s)\log(q+s). \]
	Taking $\exp$ we get
	\[ \exp(\phi(q,r,s))=\frac{q^qs^s}{(1-q)^{1-q+s}(1-s)^{1+q-s}(q+s)^{2q+2s}}.\]
	So if we want to show that $\phi(q,r,s)\leq \log(2)$, it is enough to show that the above is at most $2$. However note that even if $q, s$ were rational numbers, the above would not necessarily be rational. We again use a simple trick of raising to a power. Suppose that $N\in \Z_{>0}$ is a natural number. Then $\phi(q,r,s)\leq \log(2)$ iff
	\[ q^{Nq}s^{Ns}\leq 2^N(1-q)^{N(1-q+s)}(1-s)^{N(1+q-s)}(q+s)^{2N(q+s)}. \]
	Now if $q, s$ are rationals and $N$ is such that $Nq,Ns\in \Z$, then both sides of the above would be rational numbers. So we can potentially compute both sides exactly and compare them.
	
	Combining all the ideas, suppose that $N$ is a natural number and let $\epsilon=1/N$. Suppose that $\qbar=i/N$ and $\sbar=j/N$ for some $i, j\in \set{0, \dots,N-1}$ and we want to certify that $\phi(q,1-q-s, s)\leq \log(2)$ for all $(q, s)\in [\qbar,\qbar+\epsilon]\times [\sbar, \sbar+\epsilon]$. Then it would suffice to check the following
	\begin{equation}\label{eq:verify} (\qbar+\epsilon)^i (\sbar+\epsilon)^j\leq 2^N (1-\qbar-\epsilon)^{N-i+j+1}(1-\sbar-\epsilon)^{N+i-j+1}(\qbar+\sbar)^{2(i+j+2)}, \end{equation}
	where we have upper bounded the terms on the l.h.s.\ and lower bounded the terms on the r.h.s. All numbers appearing as the base of the exponent are at most $1$, so on the l.h.s.\ we want the minimum possible exponents and on the r.h.s.\ the opposite.
	
	Our goal is to cover $U$ with patches of the form $[\qbar,\qbar+\epsilon]\times[\sbar,\sbar+\epsilon]$ and for each such $\qbar, \sbar$ verify \cref{eq:verify}. One might hope that taking $N$ large enough would give us enough room to absorb the approximation errors. Unfortunately there is a corner case that never gets resolved for any $N$; \cref{eq:verify} will never be satisfied for $\qbar, \sbar=0$, as the r.h.s.\ will be $0$. The problematic term is the last one: $(\qbar+\sbar)^{2(i+j+2)}$. To resolve this, note that $x^x$ is a monotone decreasing function of $x$ as long as $x<1/e$. So if $\qbar+\sbar+2\epsilon<1/e$, then we can replace the problematic term with $(\qbar+\sbar+2\epsilon)^{2(i+j+2)}$ which is nonzero. The condition $\qbar+\sbar+2\epsilon<1/e$ is automatically satisfied for $\qbar=\sbar=0$ when $N>2e$.
	
	The algorithm described above can be seen in \cref{alg:verify}. The arithmetic has to be done exactly by using fractional numbers. We implemented the algorithm in Python 3, which supports infinite precision arithmetic with fractional numbers natively, and found that for $N=2000$ all tests pass successfully. The implementation can be found in \cref{alg:python}, which finished running in a few hours on a MacBook Pro computer.
\end{proof}

%% file: bp.tex
\section{Upper bound based on fractional belief propagation}
\label{sec:bp}

In this section, we prove \cref{thm:bp}. Our goal here is to prove that there exists a doubly stochastic matrix $P$, for which
\[ \per(A)\leq \prod_{i,j}(A_{i,j}/P_{i,j})^{P_{i,j}}\sqrt{(1-P_{i,j})^{1-P_{i,j}}}. \]
We remark that this upper bound is again tight for the example mentioned at the end of \cref{sec:bethe}. Naturally, we prove that this inequality holds when $P$ is the same matrix used in \cref{sec:upperbound}, the marginal matrix of the natural distribution $\mu$ on perfect matchings. Taking $\log$ from both sides of \cref{thm:bp}, we will show the following statement. 
\begin{lemma}
	With the above definitions, we have
	\[ \log(\per(A))\leq \sum_{e\in [n]\times [n]} \parens*{P_e\log(A_e/P_e)+\frac{1}{2}(1-P_e)\log(1-P_e)}. \] 
\end{lemma}
\begin{proof}
	Our starting point is the permanent upper bound we have already derived in \cref{sec:upperbound}, namely \cref{eq:entropyub}:
	\[ \log(\per(A))\leq \sum_{e\in [n]\times [n]}P_e\log(A_e/P_e)+\sum_{i\in [n]}\sum_{k\in [n]}P_{i, k}\cdot \ExX*{\pi\sim S_n}{\log\sum_{j\geq_{\pi} k} P_{i, j}}. \]
	Comparing this to the statement we would like to prove, it is enough to show the following
	\[ \frac{1}{2}\sum_{e\in [n]\times [n]} (1-P_e)\log(1-P_e)\geq \sum_{i\in [n]}\sum_{k\in [n]}P_{i,k}\ExX*{\pi\sim S_n}{\log \sum_{j\geq_\pi k}P_{i,j}}. \]
	Note that we have, once again, removed all references to the matrix $A$, and now need to prove a statement about the doubly stochastic matrix $P$. We simplify further by decomposing the inequality in terms of the rows of $P$. We show that for each $i\in [n]$,
	\[ \frac{1}{2}\sum_{k\in [n]}(1-P_{i,k})\log(1-P_{i,k})\geq \sum_{k\in [n]}P_{i, k}\ExX*{\pi\sim S_n}{\log \sum_{j\geq_\pi k} P_{i,j}}. \]
	This is a statement about the $i$-th row of $P$, and thus can be abstracted away in terms of a stochastic vector $(p_1,\dots,p_n)\in \Delta_n$. We will show that for any such vector we have
	\[  \sum_{k\in [n]} (1-p_k)\log(1-p_k)\geq 2\ExX*{\pi\sim S_n}{\sum_{k\in [n]} p_k{\sum_{j\geq_\pi k}p_j}}. \]
	We would like to remove the averaging over $\pi$ from this statement; we do this by replacing the average on the right by the sum of the expression for $\pi$ and its reverse $\overline{\pi}$. Without loss of generality, we take $\pi$ to be the identity permutation and prove that
	\[ \sum_{k\in [n]}(1-p_k)\log(1-p_k)\geq \sum_{k\in [n]}p_k\log(p_1+\dots+p_k)+\sum_{k\in [n]}p_k\log(p_k+\dots+p_n). \]
	This will be the content of \cref{lem:bprem}. Proving it finishes the proof for \cref{thm:bp}.
\end{proof}

\begin{lemma}\label{lem:bprem}
	For any sotchastic vector $(p_1,\dots,p_n)\in \Delta_n$, we have
	\[ \sum_{k\in [n]}(1-p_k)\log(1-p_k)\geq \sum_{k\in [n]}p_k\log(p_1+\dots+p_k)+\sum_{k\in [n]}p_k\log(p_k+\dots+p_n). \]
\end{lemma}
\begin{proof}
	We will rewrite each sum on the r.h.s. We have
	\[ \sum_{k\in [n]}p_k\log(p_1+\dots+p_k)=\sum_{k\in [n]}(p_1+\dots+p_{k-1})\log\parens*{\frac{p_1+\dots+p_{k-1}}{p_1+\dots+p_{k}}}. \]
	Note that the term inside the sum is taken to be $0$ for $k=1$, because we generally take $x\log(1/x)$ evaluated at $0$ to be its limit at $x=0$, which is $0$. Similarly we have
	\[ \sum_{k\in [n]}p_k\log(p_k+\dots+p_n)=\sum_{k\in [n]}(p_{k+1}+\dots+p_n)\log\parens*{\frac{p_{k+1}+\dots+p_n}{p_{k}+\dots+p_{n}}}.  \]
	We will now prove our desired inequality term-by-term. For each $k$, we will prove that
	\begin{multline*}
		(1-p_k)\log(1-p_k)\geq \\ (p_1+\dots+p_{k-1})\log\parens*{\frac{p_1+\dots+p_{k-1}}{p_1+\dots+p_{k}}}+(p_{k+1}+\dots+p_n)\log\parens*{\frac{p_{k+1}+\dots+p_n}{p_{k}+\dots+p_{n}}}.
	\end{multline*}
	Let $q=p_1+\dots+p_{k-1}$ and $r=p_{k+1}+\dots+p_n$. Then it is enough to prove that
	\[ (1-p_k)\log(1-p_k)\geq q\log(q/(q+p_k))+r\log(r/(r+p_k)). \]
	Note that the function $f(x):=x\log(x/(x+c))$ is convex over $\R_{\geq 0}$, for any fixed constant $c\geq 0$; the second derivative, $\partial^2 f/\partial x^2$, is simply $c^2/x(c+x)^2>0$. This means that as we fix $p_k$ and vary $q$ and $r$, the r.h.s.\ of the above becomes a convex function over the line segment $\set{(q, r)\given q+r=1-p_k, q\geq 0, r\geq 0}$. So the r.h.s.\ is maximized at one of the endpoints of this line segment, that is when $\set{q, r}= \set{0, 1-p_k}$. But, it is easy to see that the r.h.s.\ at either of the two endpoints is simply equal to $(1-p_k)\log(1-p_k)$, which is the l.h.s.
\end{proof}

To finish this section, we prove that $\bp_{-1/2}(A)$ is a $\sqrt{e}^n$ approximation to $\per(A)$.
\begin{theorem}
	For any matrix $A\in \R_{\geq 0}^{n\times n}$, we have
	\[ \bp_{-1/2}(A)\geq \per(A)\geq \sqrt{e}^{-n}\bp_{-1/2}(A). \]
\end{theorem}
\begin{proof}
	We have already shown that $\bp_{-1/2}(A)\geq \per(A)$. To prove the other inequality, we note that, as proved by \textcite{Gur11}, $\bethe(A)=\bp{-1}(A)$ is a lower bound on $\per(A)$. Thus, it is enough to prove that
	\[  \bp_{-1}(A)\geq \sqrt{e}^{-n}\bp_{-1/2}(A). \]
	To show this, it is enough to prove that for any doubly stochastic matrix $P$, the ratio of the terms in \cref{def:bp} for $\gamma=-1$ and $\gamma=-1/2$ is at most $\sqrt{e}^{-n}$; that is
	\[ \sqrt{\prod_{e\in [n]\times [n]}{(1-P_e)^{1-P_e}}}\geq \sqrt{e}^{-n}. \]
	But this follows from the fact that $(1-P_e)^{1-P_e}\geq e^{-P_e}$ for each $e$, and that $\sum_{e\in [n]\times [n]}P_e=n$.
\end{proof}